\documentclass[twocolumn]{autart}
\usepackage{}    
\usepackage{subfigure}
\usepackage{amsfonts}
\usepackage{mathrsfs}
\usepackage[dvips]{color}
\usepackage{amsmath,amssymb}
\usepackage{natbib}
\usepackage{graphicx}          
\def\dref#1{(\ref{#1})}
\newtheorem{lemma}{Lemma}
\newtheorem{assumption}{Assumption}
\newtheorem{theorem}{Theorem}
\newtheorem{corollary}{Corollary}
\newtheorem{remark}{Remark}
\newtheorem{problem}{Problem}
\newtheorem{example}{Example}

\def\prooft{\noindent{\bf Proof of Theorem 3} }
\def\proofthm{\noindent{\bf Proof of Theorem 4} }

\UseRawInputEncoding
\begin{document}

\begin{frontmatter}

\title{Minimal-order Appointed-time Unknown Input Observers: Design and Applications} 

\author[lv]{Yuezu Lv}\ead{yzlv@seu.edu.cn},  
\author[duan]{Zhongkui Li}\ead{zhongkli@pku.edu.cn},    
\author[duan]{Zhisheng Duan}\ead{duanzs@pku.edu.cn}  


\address[lv]{Department of Systems Science, School of Mathematics, Southeast University, Nanjing 211189, China} %
\address[duan]{State Key Laboratory for Turbulence and
Complex Systems, Department of Mechanics and Engineering Science,
College of Engineering, Peking University, Beijing 100871, China}

\begin{keyword}                           
Appointed-time unknown input observer, minimal-order observer, attack-free protocol, consensus                      
\end{keyword}                             

\begin{abstract}                          

This paper presents a framework on minimal-order appointed-time unknown input observers for linear systems based on the pairwise observer structure. A minimal-order appointed-time observer is first proposed for the linear
system without the unknown input, which can estimate the state exactly at the preset time by seeking for the unique solution of a system of linear equations. To further release the computational burden, another form of the appointed-time observer is designed. For the general linear system with the unknown input acting on both the system dynamics and the measured output, the model reconfiguration is made to decouple the effect of the unknown input, and the gap between the existing reduced-order appointed-time unknown input observer and the possible minimal-order appointed-time observer is revealed. Based on the reconstructed model, the minimal-order appointed-time unknown input observer is presented to realize state estimation of linear system with the unknown input at the arbitrarily small preset time. The minimal-order appointed-time unknown input observer is then applied to the design of fully distributed adaptive output-feedback attack-free consensus protocols for linear multi-agent systems.
\end{abstract}

\end{frontmatter}

\section{Introduction}\label{s1}
The state estimation has been well investigated since the invention of the well-known Kalman filter [\cite{kalman1960jbe}] and Luenberger observer [\cite{luenberger1964tme}] in 1960s,
and various observers have been developed for linear or nonlinear systems [\cite{houm2002tac,dezaf1993tac,dingz2012tac}].
The unknown input observer is a typical state estimation for the systems with external unknown inputs, which has attracted widespread attentions due to its resultful applications in the fields of fault diagnosis [\cite{gaoz2016tie,cristofaro2014auto}] and attack detection [\cite{amins2013tcst,amelia2018tps}].

The generalized dynamic model of linear time-invariant system with unknown input can be formulated by
\begin{equation}\label{model}
\begin{aligned}
\dot{x}&=Ax+Bu+Ew,\\
y&=Cx+Du+Fw,
\end{aligned}
\end{equation}
where $x\in\mathbf{R}^n$, $u\in\mathbf{R}^p$, $w\in\mathbf{R}^q$, $y\in\mathbf{R}^m$ are the state, the control input, the unknown input and the measurement output, respectively. In practice, the unknown input $w$ can represent the external disturbances, unmodeled dynamics or actuator failures.

The unknown input observer has been investigated from various perspectives. A procedure of the minimal-order unknown input observer was proposed in \cite{wang1975tac} for the linear system \dref{model} with $F=0$, and the existence conditions were revealed in \cite{kudva1980tac} that the rank of $CE$ equals to that of $E$ and the triple $(A,E,C)$ has stable or even no invariant zeros. Following the conditions proposed in \cite{kudva1980tac}, full-order unknown input observers were designed in \cite{yang1988tac,darouach1994tac}, and the reduced-order observers were presented in \cite{houm1992tac,syrmos1993tac}. \cite{kurek1983tac,houm1994tac} further studied the general model \dref{model}, and presented minimal-order observer design procedure as well as the existence conditions. Unknown input observers for discrete-time systems were illustrated in \cite{valcher1999tac,sundaram2007tac,sundaram2008auto}, and \cite{darouach1996tac,koenig2002tac,zhang2020tac} investigated the unknown input observer design for descriptor systems. The unknown input functional observers were designed in \cite{sundaram2008auto,trinh2008tac,sakhraoui2020tac}. Unknown input observers for the switched systems [\cite{bejarano2011tac,zhang2020tac}] and $h_\infty$ unknown input observers [\cite{gao2016auto}] have also been well studied.

One common feature of the aforementioned unknown input observers is that the system state is estimated asymptotically. In practical applications such as the fault detection, it is desired to realize finite-time estimation of the state. Among all the categories of finite-time convergence, the strictest one is to reach convergence exactly at the preset time instant, which is named as appointed-time or specified-time convergence [\cite{zhaoy2019tac}]. The appointed-time observer for linear systems without the unknown input was proposed in \cite{engel2002tac}, where a pairwise observer structure was designed, consisting of two Luenberger observers and achieving the appointed-time state estimation based on time-delayed observer information. By introducing a time-varying coordinate transformation matrix, a novel observer for linear systems was designed in \cite{pingilberto2020cdc}, which successfully realized the appointed-time state estimation with an arbitrarily small predetermined time. Based on the pairwise observer structure,
the appointed-time observers for nonlinear systems were presented in \cite{kreisselmeier2003tac,menold2003cdc}, and the appointed-time functional observers for linear systems were studied in \cite{raff2005cdc}. \cite{liy2015auto} further considered the appointed-time state estimation of nonlinear systems with measurement noise. The appointed-time observer for discrete-time systems was presented in \cite{ao2018is}, where the applications on the attack detection were also investigated. Following the observer design structure of \cite{engel2002tac}, the appointed-time unknown input observer for linear system \dref{model} with $F=0$ was proposed in \cite{rafft2006mcca}. Distributed appointed-time unknown input observers were further investigated in \cite{lvyz2020tcns}, based on which fully distributed attack-free consensus protocols were proposed for multi-agent systems.

Notice that the above-mentioned appointed-time observers based on the pairwise observer structure are either of full order $2n$, or of reduced order $2(n-\text{rank}(E))$. From the point of view of realization, it is favourable to design minimal-order appointed-time observers, which is expected to be of order $2(n-\text{rank}(C))$ when $F=0$. In this paper, we intend to answer whether such minimal-order appointed-time observer exists and how to design the observers.

For the linear system without the unknown input, we first give a thorough analysis of the pairwise observer design structure presented in \cite{engel2002tac} to reveal how it works on realizing state estimation at the appointed time. That is, to build a system of $6n$ linear equations in $6n$ unknowns, and construct the observer expression based on the unique solution of the system of linear equations. Following such design methodology, the pairwise minimal-order observers with different poles are proposed, and a system of $(6n-4m)$ equations in $(6n-4m)$ unknowns is constructed by adding the $2m$ equations of measured output at time instant $t$ as well as the delayed time instant $t-\tau$. It is demonstrated that the coefficient matrix is invertible, which gives a unique solution to the system of linear equations. The appointed-time observer is then designed by taking the portion of the unique solution. To release the computation burden caused by calculating the inverse of the high-dimensional coefficient matrix, another form of the minimal-order appointed-time observer is formulated, whose structure is coincident with that of the full-order appointed-time observer in \cite{engel2002tac}.

For the linear system with the unknown input, we first reconstruct the model to decouple the effect of the unknown input, and exhibit both full-order and reduced-order appointed-time unknown input observers based on different reconstructed models. The gap between the reduced-order and expected minimal-order appointed-time unknown input observers is revealed, which motivates us to further decrease the observer order. Following the observer design structure of the minimal-order appointed-time observer for linear systems without the unknown input, the minimal-order appointed-time unknown input observer is obtained by designing the observer to estimate the state of the reconstructed model at the appointed time. The special case that the unknown input does not act on the measured output, i.e., $F=0$, is also discussed. The proposed minimal-order appointed-time unknown input observer is then applied into the consensus problem of linear multi-agent systems, where distributed minimal-order appointed-time unknown input observer is put forward to estimate the consensus error by viewing the relative input among neighboring agents as the unknown input, and the distributed adaptive attack-free consensus protocol is presented based on the consensus error estimation. The proposed protocol possesses the feature of avoiding information transmission via communication channel, which takes the advantages of reducing the communication cost and being free from network attacks.

The rest of this paper is organized as follows. Section \ref{s2} presents the design structures of the minimal-order appointed-time observer for linear system \dref{model} without unknown input $w$.
Section \ref{s3} further studies the minimal-order appointed-time unknown input observers. Section \ref{s4} applies the appointed-time unknown input observer into the design of fully distributed adaptive attack-free consensus protocols for linear multi-agent systems, and gives a simulation example to illustrate the effectiveness of the proposed methods. Section \ref{s6} concludes this paper.

Notations: Let $\mathbb{R}^{n\times m}$ be the set of $n\times m$ matrices. $I_p$ represents the $p$-dimensional identity matrix.
Symbol $\text{diag}(x_1,\cdots,x_n)$ represents a diagonal matrix with diagonal elements being $x_i$. For a matrix $A$, $A^+$ denotes its generalized inverse with $AA^+A=A$ and $A^+AA^+=A^+$.
For a square matrix $Z$,
$\Re_i\{\lambda(Z)\}$ represents the real part of the $i$-th eigenvalue of $Z$.

\section{Minimal-order Appointed-time Observers}\label{s2}
We first study the minimal-order appointed-time observer for linear systems without the unknown input, i.e., $w=0$, or $E=0,F=0$ in model \dref{model}. Without loss of generality, we assume that $C$ is of full row rank.

\subsection{Problem Analysis}
Under the assumption that $(A,C)$ is observable, the full-order appointed-time observer was given in \cite{engel2002tac} as
\begin{equation}\label{obs10}
\begin{aligned}
\dot{\bar {v}}_1&=(A+L_1C)v_1-L_1y+(B+L_1D)u,\\
\dot{\bar {v}}_2&=(A+L_2C)v_2-L_2y+(B+L_2D)u,\\
\bar x(t)&=\bar D_c[\bar v(t)-e^{\bar A_c\tau}\bar v(t-\tau)],
\end{aligned}
\end{equation}
where $\bar v=[\bar v_1^T,\bar v_2^T]^T$, $\bar A_c=\text{diag}(A+L_1C,A+L_2C)$ with $L_1$ and $L_2$ as the gain matrices satisfying $\Re_i\{\lambda(A+L_2C)\}<\sigma<\Re_j\{\lambda(A+L_1C)\}<0,\forall i,j=1,\cdots,n$, $\sigma$ is a negative constant, $\bar D_c=\begin{bmatrix}I_n&0\end{bmatrix}\begin{bmatrix}\bar C_c&e^{\bar A_c\tau}\bar C_c\end{bmatrix}^{-1}$ with $\bar C_c=\begin{bmatrix}I_n\\I_n\end{bmatrix}$, and $\tau$ is a positive constant.

The methodology of observer \dref{obs10} is to construct a system of linear equations, which has a unique solution [\cite{engel2002tac}]. Specifically, define $\tilde{\bar{v}}_i=\bar v_i-x,i=1,2,$ and $\tilde{\bar{v}}=[\tilde{\bar{v}}_1^T,\tilde{\bar{v}}_2^T]^T$. Then, we have
\begin{equation}\label{lineareq0}
\begin{aligned}
&\bar C_cx(t)+\tilde{\bar{v}}(t)=\bar v(t),\\
&\bar C_cx(t-\tau)+\tilde{\bar{v}}(t-\tau)=\bar v(t-\tau),\\
&\tilde{\bar{v}}(t)=e^{\bar A_c\tau}\tilde{\bar{v}}(t-\tau),
\end{aligned}
\end{equation}
which contains $6n$ linear equations in $6n$ unknowns $x(t),x(t-\tau),\tilde{\bar{v}}(t),\tilde{\bar{v}}(t-\tau)$. The unique solution of $x(t)$ can be calculated as $\bar x(t)$ in \dref{obs10}. Thus, it is concluded in \cite{engel2002tac} that the observer \dref{obs10} can estimate the exact value of $x(t)$ at appointed time $\tau$, i.e., $\bar x(t)\equiv x(t),\forall t\geq\tau$.

To design minimal-order appointed-time observer, we can only construct two $(n-m)$-order observers:
\begin{equation}\label{obs11}
\begin{aligned}
\dot{{v}}_1&=M_1v_1+H_1y+(T_1B-H_1D)u,\\
\dot{{v}}_2&=M_2v_2+H_2y+(T_2B-H_2D)u,
\end{aligned}
\end{equation}
where $M_1,M_2\in\mathbb{R}^{(n-m)\times(n-m)}$ are gain matrices satisfying $\Re_i\{\lambda(M_2)\}<\sigma<\Re_j\{\lambda(M_1)\}<0,\forall i,j=1,\cdots,n-m$, $H_1$ and $H_2$ are matrices such that both $(M_1,H_1)$ and $(M_2,H_2)$ are controllable, $T_1$ and $T_2$ are respectively the unique solutions of the Sylvester equations
\begin{equation}\label{syl}
T_iA-M_iT_i=H_iC,~i=1,2,
\end{equation}
such that $\begin{bmatrix}T_i\\C\end{bmatrix}$ are invertible. Let $U_i{=}[S_i~\bar S_i]{=}\begin{bmatrix}T_i\\C\end{bmatrix}^{-1}$.

It is well-known that each $v_i$ can exponentially converge to $T_ix$. Define $\tilde{v}_i=v_i-T_ix$, and we have
\begin{equation}\label{lineareq}
\begin{aligned}
&T_ix(t)+\tilde{{v}}_i(t)=v_i(t),\\
&T_ix(t-\tau)+\tilde{{v}}_i(t-\tau)=v_i(t-\tau),\\
&\tilde{{v}}_i(t)=e^{M_i\tau}\tilde{{v}}_i(t-\tau),~~i=1,2.
\end{aligned}
\end{equation}
It is clear that \dref{lineareq} contains $6(n-m)$ equations in $2n+4(n-m)$ unknowns $x(t),x(t-\tau),\tilde{v}_i(t),\tilde{v}_i(t-\tau)$. As a result, $x(t)$ cannot be uniquely determined by the above system of linear equations. To design minimal-order appointed-time observer, the main difficulty lies in constructing an appropriate system of linear equations to calculate a unique solution of $x(t)$.

\subsection{Observer Design}
To make the system of linear equations \dref{lineareq} have a unique solution, $2m$ extra independent equations are required. It is natural to add the following $2m$ equations:
\begin{equation}\label{leq0}
\begin{aligned}
&Cx(t)=y(t)-Du(t),\\
&Cx(t-\tau)=y(t-\tau)-Du(t-\tau),
\end{aligned}
\end{equation}
and the system of linear equations \dref{lineareq} and \dref{leq0} can be written as
\begin{equation}\label{leq1}A_0x_0=b_0\end{equation}
with $x_0=[x^T(t),\tilde{v}_1^T(t),\tilde{v}_2^T(t), x^T(t-\tau),\tilde{v}_1^T(t-\tau),\tilde{v}_2^T(t-\tau)]^T$,
$b_0=[v_1^T(t),(y(t)-Du(t))^T,0_{n-m}^T,v_2^T(t), v_1^T(t-\tau),(y(t-\tau)-Du(t-\tau))^T,0_{n-m}^T,v_2^T(t-\tau)]^T$ and
$$
A_0=\begin{bmatrix}T_1&I_{n-m}&0&0&0&0\\C&0&0&0&0&0\\0&I_{n-m}&0&0&-e^{M_1\tau}&0\\
T_2&0&I_{n-m}&0&0&0\\
0&0&0&T_1&I_{n-m}&0\\0&0&0&C&0&0\\0&0&I_{n-m}&0&0&-e^{M_2\tau}\\0&0&0&T_2&0&I_{n-m}
\end{bmatrix}.
$$

We now have $2n+4(n-m)$ equations in $2n+4(n-m)$ unknowns. The following result shows that the system of linear equations \dref{leq1} has a unique solution.
\begin{theorem}\label{theorem1}
Suppose that $(A,C)$ is observable. Then $A_0$ is invertible for almost all $\tau>0$. Moreover, the observer
\begin{equation}\label{obs13}
\hat x(t)=\begin{bmatrix}I_n&0\end{bmatrix}A_0^{-1}b_0
\end{equation}
can estimate the state $x(t)$ of the linear system \dref{model} without the unknown input $w$ at appointed time $\tau$.
\end{theorem}
\begin{proof}
We only have to show the invertibility of $A_0$, since the system of linear equations \dref{leq1} has a unique solution if and only if $A_0$ is invertible, which gives $x(t)=\begin{bmatrix}I_n&0\end{bmatrix}A_0^{-1}b_0$ for $t\geq\tau$.

The invertibility of $A_0$ is equivalent to that the determinant of $A_0$ is nonzero, i.e., $|A_0|\neq0$. Let
$$P_0=\begin{bmatrix}I_n&-S_1&0&0&-S_1e^{M_1\tau}&0\\0&I_{n-m}&0&0&e^{M_1\tau}&0\\
0&0&I_{n-m}&0&T_2S_1e^{M_1\tau}&0\\
0&0&0&I_n&-S_1&0\\0&0&0&0&I_{n-m}&0\\0&0&0&0&0&I_{n-m}
\end{bmatrix},$$
and we have
\begin{equation*}
\begin{aligned}
&A_0P_0=
\begin{bmatrix}U_1^{-1}&0&0&0&0&0\\0&I_{n-m}&0&0&0&0\\
T_2&-T_2S_1&I_{n-m}&0&0&0\\
0&0&0&U_1^{-1}&0&0\\0&0&I_{n-m}&0&T_2S_1e^{M_1\tau}&-e^{M_2\tau}\\0&0&0&T_2&-T_2S_1&I_{n-m}
\end{bmatrix}.
\end{aligned}
\end{equation*}
We can obtain that
\begin{equation*}
\begin{aligned}
|A_0|&=|A_0P_0|=|U_1^{-1}|^2\begin{vmatrix}T_2S_1e^{M_1\tau}&-e^{M_2\tau}\\-T_2S_1&I_{n-m}\end{vmatrix}\\
&=|U_1^{-1}|^2\begin{vmatrix}T_2S_1e^{M_1\tau}-e^{M_2\tau}T_2S_1&0\\-T_2S_1&I_{n-m}\end{vmatrix}\\
&=|U_1^{-1}|^2|e^{M_1\tau}|\begin{vmatrix}T_2S_1-e^{M_2\tau}T_2S_1e^{-M_1\tau}\end{vmatrix}.
\end{aligned}
\end{equation*}

Note that both $\begin{bmatrix}T_i\\C\end{bmatrix},i=1,2$ are invertible. Let $\begin{bmatrix}T_2\\C\end{bmatrix}=\begin{bmatrix}P_1&P_{12}\\0&I_m\end{bmatrix}\begin{bmatrix}T_1\\C\end{bmatrix}$, and we have that $\text{rank}(P_1)=n-m$. Thus, $$|T_2S_1|=\left|\begin{bmatrix}P_1&P_{12}\end{bmatrix}\begin{bmatrix}T_1\\C\end{bmatrix}S_1\right|=|P_1|\neq0.$$

On the other hand, $$e^{M_2\tau}T_2S_1e^{-M_1\tau}=e^{(M_2-\sigma I_{n-m})\tau}T_2S_1e^{(\sigma I_{n-m}-M_1)\tau}.$$
Since $\Re_i\{\lambda(M_2)\}<\sigma<\Re_j\{\lambda(M_1)\}$, we can obtain that $e^{M_2\tau}T_2S_1e^{-M_1\tau}\rightarrow0$ as $\tau\rightarrow\infty$, which implies that $|T_2S_1-e^{M_2\tau}T_2S_1e^{-M_1\tau}|\rightarrow|T_2S_1|$ as $\tau\rightarrow\infty$. Thus, the overall determinant $|A_0|\neq0$ with sufficiently large $\tau$. Since the determinant is an analytic function of $\tau$, it only has isolated zeros. Notice that $|T_2S_1-e^{M_2\tau}T_2S_1e^{-M_1\tau}|=0$ when $\tau=0$. Therefore, for almost all $\tau>0$, $|A_0|\neq 0$. This completes the proof. $\hfill\blacksquare$
\end{proof}
\begin{remark}
Theorem \ref{theorem1} reveals that the minimal-order appointed-time observers do exist if $(A,C)$ is observable. Since $|A_0|$ has only isolated zeros and $|A_0|=0$ when $\tau=0$, there exists $\tau^*>0$ such that $|A_0|\neq0,\forall \tau\in(0,\tau^*)$, meaning that the predetermined time $\tau$ can be arbitrarily small. The methodology of the observer is to design two $(n-m)$-order observers, and introduce delayed information of the observers as well as the measurement output to generate a system of $2n+4(n-m)$ linear equations in $2n+4(n-m)$ variables, where the coefficient matrix $A_0$ is invertible so that there is a unique solution for the system of linear equations. Compared to the existing works on designing appointed-time observer with the pairwise observer structure [\cite{engel2002tac,kreisselmeier2003tac,menold2003cdc,liy2015auto}], the observer \dref{obs13} is in a minimal-order form with the design of $2(n-m)$-order dynamical variables $v_1$ and $v_2$, which has the advantage of consuming lower computation cost.
\end{remark}

\subsection{Observer Reconstruction}
It should be noticed that the coefficient matrix $A_0$ is of $(6n-4m)\times(6n-4m)$ dimension, and the inverse of $A_0$ may not be easy to calculate when $n$ increases. Besides, the design of the observer \dref{obs13} is quite different from that of the full-order observer $\bar x(t)$ in \dref{obs10}. In this subsection, we intend to reconstruct the appointed-time observer with the information of $b_0$.

Define $\phi=[v_1^T,(y-Du)^T,v_2^T,(y-Du)^T]^T$. Choose $\bar M_1$ and $\bar M_2$ such that $\Re_i\{\lambda(\bar M_2)\}<\sigma<\Re_j\{\lambda(\bar M_1)\}<0,\forall i,j=1,\cdots,m$.
Let $D_c=[I_n~0]\begin{bmatrix}I_n&U_1e^{\hat M_1\tau}U_1^{-1}\\I_n&U_2e^{\hat M_2\tau}U_2^{-1}\end{bmatrix}^{-1}$ with $\hat M_i=\text{diag}(M_i,\bar M_i),i=1,2$. We can redesign the minimal-order appointed-time observer by
\begin{equation}\label{obs12}
\hat x(t)=D_c\begin{bmatrix}U_1&~\\~&U_2\end{bmatrix}[\phi(t)-e^{\hat M\tau}\phi(t-\tau)],
\end{equation}
where $\hat M=\text{diag}(\hat M_1,\hat M_2)$.
\begin{theorem}\label{theorem2}
Suppose that $(A,C)$ is observable. For almost all $\tau>0$, the observer \dref{obs12} can estimate the state $x(t)$ of the linear system \dref{model} without the unknown input $w$ at appointed time $\tau$.
\end{theorem}
\begin{proof}
Note that $$v_i(t)-e^{M_i\tau}v_i(t-\tau)=T_ix(t)-e^{M_i\tau}T_ix(t-\tau).$$
We can calculate that
\begin{equation*}
\begin{aligned}
&\phi(t)-e^{\hat M\tau}\phi(t-\tau)
=\begin{bmatrix}U_1^{-1}x(t)-e^{\hat M_1\tau}U_1^{-1}x(t-\tau)\\
U_2^{-1}x(t)-e^{\hat M_2\tau}U_2^{-1}x(t-\tau)\end{bmatrix}.
\end{aligned}
\end{equation*}
And we have
\begin{equation*}
\begin{aligned}
\hat{x}(t)=&D_c\bar C_cx(t)-D_c
\begin{bmatrix}U_1e^{\hat M_1\tau}U_1^{-1}\\
U_2e^{\hat M_2\tau}U_2^{-1}\end{bmatrix}x(t-\tau)\\
=&x(t),
\end{aligned}
\end{equation*}
where the last equality is obtained by noticing that $D_c\bar C_c=I_n$ and $D_c
\begin{bmatrix}U_1e^{\hat M_1\tau}U_1^{-1}\\
U_2e^{\hat M_2\tau}U_2^{-1}\end{bmatrix}=0$.

The remaining is to show the existence of $D_c$, or the invertibility of the matrix $\begin{bmatrix}I_n&U_1e^{\hat M_1\tau}U_1^{-1}\\I_n&U_2e^{\hat M_2\tau}U_2^{-1}\end{bmatrix}$, which is equivalent to demonstrate $\begin{vmatrix}I_n&U_1e^{\hat M_1\tau} U_1^{-1}\\I_n&U_2e^{\hat M_2\tau} U_2^{-1}\end{vmatrix}\neq0$. Since
\begin{equation*}
\begin{aligned}
&\begin{bmatrix}I_n&U_1e^{\hat M_1\tau} U_1^{-1}\\I_n&U_2e^{\hat M_2\tau} U_2^{-1}\end{bmatrix}\\
=&\begin{bmatrix}I_n&0\\I_n&-I_n\end{bmatrix}\begin{bmatrix}I_n&U_1e^{\hat M_1\tau} U_1^{-1}\\0&U_1e^{\hat M_1\tau} U_1^{-1}-U_2e^{\hat M_2\tau} U_2^{-1}\end{bmatrix},
\end{aligned}
\end{equation*}
we have that
\begin{equation*}
\begin{aligned}
&\begin{vmatrix}I_n&U_1e^{\hat M_1\tau} U_1^{-1}\\I_n&U_2e^{\hat M_2\tau} U_2^{-1}\end{vmatrix}=(-1)^n\left|U_1e^{\hat M_1\tau}U_1^{-1}-U_2e^{\hat M_2\tau}U_2^{-1}\right|\\
=&(-1)^n\left|U_1e^{\hat M_1\tau}U_1^{-1}\right|
\left|I_n-U_1e^{-\hat M_1\tau}U_1^{-1}U_2e^{\hat M_2\tau}U_2^{-1}\right|.
\end{aligned}
\end{equation*}
On the one hand, we have
$$\left|U_1e^{\hat M_1\tau}U_1^{-1}\right|=\left|e^{\hat M_1\tau}\right|\neq0.
$$
On the other hand, note that
\begin{equation*}
\begin{aligned}
&U_1e^{-\hat M_1\tau}U_1^{-1}U_2e^{\hat M_2\tau}U_2^{-1}\\
=&U_1e^{(\sigma I_n-\hat M_1)\tau}U_1^{-1}U_2e^{(\hat M_2-\sigma I_n)\tau}U_2^{-1}.
\end{aligned}
\end{equation*}
Since $\Re_i\{\lambda(\hat M_2)\}<\sigma<\Re_j\{\lambda(\hat M_1)\},\forall i,j=1,\cdots,n$, we can obtain that $U_1e^{-\hat M_1\tau}U_1^{-1}U_2e^{\hat M_2\tau}U_2^{-1}\rightarrow0$ as $\tau\rightarrow\infty$, which in turn implies that as $\tau\rightarrow\infty$, $$\left|I_n-U_1e^{-\hat M_1\tau}U_1^{-1}U_2e^{\hat M_2\tau}U_2^{-1}\right|\rightarrow 1.$$ Thus, the determinant $\begin{vmatrix}I_n&U_1e^{\hat M_1\tau} U_1^{-1}\\I_n&U_2e^{\hat M_2\tau} U_2^{-1}\end{vmatrix}\neq0$ with sufficiently large $\tau$. In light of the fact that the determinant is an analytic function of $\tau$, it only has isolated zeros. Since $\begin{vmatrix}I_n&U_1e^{\hat M_1\tau} U_1^{-1}\\I_n&U_2e^{\hat M_2\tau} U_2^{-1}\end{vmatrix}=0$ when $\tau=0$, $D_c$ exists for almost all $\tau>0$.
Therefore, for almost all $\tau>0$, the minimal-order observer $\hat x$ in \dref{obs12} can estimate the state at appointed time $\tau$. $\hfill\blacksquare$
\end{proof}
\begin{remark}
Both the observers \dref{obs12} and \dref{obs13} can realize appointed-time estimation of the state of the linear system \dref{model} without the unknown input $w$. In comparison to the observer \dref{obs13} with the inverse of the $(6n-4m)\times(6n-4m)$ dimensional matrix $A_0$ to be calculated, the observer \dref{obs12} only needs to calculate the inverse of the $2n\times 2n$ dimensional matrix $\begin{bmatrix}I_n&U_1e^{\hat M_1\tau} U_1^{-1}\\I_n&U_2e^{\hat M_2\tau} U_2^{-1}\end{bmatrix}$.
\end{remark}

\subsection{Methodology of Observer \dref{obs12}}
Note that the structure of the minimal-order appointed-time observer \dref{obs12} is similar to that of the full-order observer \dref{obs10}. The intuitive explanation is given as follows.

Define $\tilde\phi=\phi-\begin{bmatrix}U_1^{-1}&~\\~&U_2^{-1}\end{bmatrix}\bar C_cx$, $\tilde{\phi}_1=[I_n~0]\tilde\phi$ and $\tilde{\phi}_2=[0~I_n]\tilde\phi$. Under the observer \dref{obs12}, we have $\tilde{\phi}_i=[\tilde{v}_i^T,\tilde\phi_{i2}^T]^T$ with $\tilde\phi_{i2}\equiv0$. Thus, the main idea of the observer \dref{obs12} is to introduce the extra $4m$ zero variables $\tilde\phi_{i2}(t)$ and $\tilde\phi_{i2}(t-\tau)$, and add the following $6m$ equations:
\begin{equation}\label{leq2}
\begin{aligned}
&Cx(t)+\tilde\phi_{i2}(t)=y(t)-Du(t),\\
&Cx(t-\tau)+\tilde\phi_{i2}(t-\tau)=y(t-\tau)-Du(t-\tau),\\
&\tilde\phi_{i2}(t)=e^{\bar M_i\tau}\tilde\phi_{i2}(t-\tau),~i=1,2,
\end{aligned}
\end{equation}
where $\bar M_i$ is chosen according to the same requirement of $M_i$.

Combining \dref{lineareq} and \dref{leq2} yields a system of $6n$ linear equations in $6n$ variables $x(t),x(t-\tau),\tilde\phi(t),\tilde\phi(t-\tau)$ as follows:
\begin{equation}\label{leq3}
\begin{aligned}
&\begin{bmatrix}U_1^{-1}&~\\~&U_2^{-1}\end{bmatrix}\bar C_cx(t)+\tilde\phi(t)=\phi(t),\\
&\begin{bmatrix}U_1^{-1}&~\\~&U_2^{-1}\end{bmatrix}\bar C_cx(t-\tau)+\tilde\phi(t-\tau)=\phi(t-\tau),\\
&\tilde\phi(t)=e^{\hat M\tau}\tilde\phi(t-\tau).
\end{aligned}
\end{equation}
Solving the above system of linear equations gives the appointed-time observer \dref{obs12}.
\begin{remark}
Both the full-order observer \dref{obs10} and the reduced-order observer \dref{obs12} is obtained by solving the solution of the system of $6n$ linear equations in $6n$ variables; see \dref{lineareq0} and \dref{leq3}. In light of the invertibility of the matrices $\begin{bmatrix}\bar C_c&e^{\bar A_c\tau}\bar C_c\end{bmatrix}$ and $\begin{bmatrix}I_n&U_1e^{\hat M_1\tau}U_1^{-1}\\I_n&U_2e^{\hat M_2\tau}U_2^{-1}\end{bmatrix}$, it is not difficult to verify that both the coefficient matrix
$\bar A_0{=}\begin{bmatrix}\bar C_c&0&I_{2n}&0\\0&\bar C_c&0&I_{2n}\\
0&0&I_{2n}&-e^{\bar A_c\tau}
\end{bmatrix}$ of the linear equations \dref{lineareq0} and the coefficient matrix $\hat A_0{=}\begin{bmatrix}U_1^{-1}&0&I_{n}&0&0&0\\U_2^{-1}&0&0&I_n&0&0\\
0&U_1^{-1}&0&0&I_{n}&0\\0&U_2^{-1}&0&0&0&I_{n}\\
0&0&I_{n}&0&-e^{\hat M_1\tau}&0\\0&0&0&I_{n}&0&-e^{\hat M_2\tau}
\end{bmatrix}$ of the linear equations \dref{leq3} are invertible, implying that both the linear equations \dref{lineareq0} and \dref{leq3} have the unique solution. Clearly,
the main distinction between observer \dref{obs10} and observer \dref{obs12} lies in the introducing of nonsingular transformation matrices $U_i$. Specifically, the observer $v_i$ and the measurement output $y-Du$ forms the variable $\phi_i=[v_i^T,(Cx)^T]^T$, which is to estimate $U_i^{-1}x$; while for the full-order observer \dref{obs10}, the variable $\bar v_i$ can realize exponential estimation of $x$. In this sense, the minimal-order appointed-time observer \dref{obs12} and the full-order appointed-time observer \dref{obs10} share the same design structure. Moreover, the appointed-time observer \dref{obs12} degenerates into the observer \dref{obs10} when $U_i=I_n$.
\end{remark}
\begin{remark}
It is worth noting that the coefficient matrix $A_0$ in linear equations \dref{leq1} has lower dimension compared with $\hat A_0$, but the expression of the observer \dref{obs13} is more complicated than the observer \dref{obs12}.
Such counterintuitive result is mainly caused by the special structure of the linear equations \dref{leq3}. Compared with the linear equations \dref{leq1}, the linear equations \dref{leq3} contains the extra $4m$ zero variables $\tilde{\phi}_{i2}(t),\tilde{\phi}_{i2}(t-\tau)$. Moreover, the measurement output equations $y=Cx+Du$ are used twice, and the relation between the zero variables $\tilde{\phi}_{i2}(t)$ and $\tilde{\phi}_{i2}(t-\tau)$ is also introduced to construct the linear equations \dref{leq3}; while the design of the extra matrices $\bar M_i$ plays a key role in deriving the observer \dref{obs12}.
\end{remark}

\section{Appointed-time Unknown Input Observers}\label{s3}
In this section, we further extend the observer \dref{obs12} into the minimal-order appointed-time observer for the linear system \dref{model} in presence of the unknown input $w$.

\subsection{Model Reconfiguration and Full-order Observers}
Since the unknown input $w$ acts upon both the dynamics and the measured output, we have to first get rid of it from the measured output. Let $\bar C=(I_m-FF^+)C$, and $$\bar y=(I_m-FF^+)(y-Du).$$ Note that
$$w=F^+(y-Cx-Du)+(I_q-F^+F)w,$$ which in turn gives
\begin{equation}\label{model1}
\begin{aligned}
\dot{x}&=\bar Ax+EF^+y+\hat Bu+\bar Ew,\\
\bar y&=\bar Cx,
\end{aligned}
\end{equation}
with $\bar A=A-EF^+C$, $\hat B=B-EF^+D$, and
$\bar E=E(I_q-F^+F)$. Let $G=I_n-\bar E(\bar C\bar E)^+\bar C$, and the following lemma is introduced.
\begin{lemma}[\cite{houm1994tac}]\label{lem1}
The two statements are equivalent:
\begin{itemize}
\item [$1^\circ$] $\text{rank}\begin{bmatrix}0&F\\F&CE\end{bmatrix}
=\text{rank}(F)+\text{rank}\begin{bmatrix}E\\F\end{bmatrix}$;
\item [$2^\circ$] $G\bar E=0$.
\end{itemize}
\end{lemma}

Let $\eta=Gx$. We have
\begin{equation*}
\begin{aligned}
&x=\eta+\bar E(\bar C\bar E)^+\bar y\\
=&\eta{+}\bar E(\bar C\bar E)^+(I_m{-}FF^+)y{-}\bar E(\bar C\bar E)^+(I_m{-}FF^+)Du.
\end{aligned}
\end{equation*}
The dynamics of $\eta$ are then given by
\begin{equation}\label{model20}
\begin{aligned}
\dot{\eta}&=G\bar A\eta+Hy+\bar Bu+G\bar Ew,\\
y_{\eta}&=\bar C\eta=(I_m-\bar C\bar E(\bar C\bar E)^+)\bar y,
\end{aligned}
\end{equation}
with $\bar B=G\hat B-G\bar A\bar E(\bar C\bar E)^+(I_m-FF^+)D$ and $H=GEF^++G\bar A\bar E(\bar C\bar E)^+(I_m-FF^+)$.
If condition $1^\circ$ in Lemma \ref{lem1} holds, the model \dref{model20} can be written as
\begin{equation}\label{model2}
\begin{aligned}
\dot{\eta}&=G\bar A\eta+Hy+\bar Bu,\\
y_{\eta}&=\bar C\eta=\hat C(y-Du),
\end{aligned}
\end{equation}
where $\hat C=(I_m-\bar C\bar E(\bar C\bar E)^+)(I_m-FF^+)$.

\begin{lemma}[\cite{houm1994tac}]\label{lem2}
Under condition $1^\circ$ in Lemma \ref{lem1}, the following two statements are equivalent:
\begin{itemize}
\item [$1^\circ$] $\text{rank}\begin{bmatrix}A-sI_n&E\\C&F\end{bmatrix}
=n+\text{rank}\begin{bmatrix}E\\F\end{bmatrix},~\forall s\in\mathbf{C}$;
\item [$2^\circ$] $(G\bar A,\bar C)$ is observable.
\end{itemize}
\end{lemma}

Note that
$$
\begin{bmatrix}\bar C&0\\0&F\end{bmatrix}
=\begin{bmatrix}I_m-FF^+\\FF^+\end{bmatrix}\begin{bmatrix}C&F\end{bmatrix}
\begin{bmatrix}I_n&0\\-F^+C&I_q\end{bmatrix},
$$
implying that $\text{rank}(\bar C)=\text{rank}\begin{bmatrix}C&F\end{bmatrix}-\text{rank}(F)$.
By Lemma \ref{lem1} and Lemma \ref{lem2}, it is not difficult to derive that we can design the asymptotically convergent unknown input observer with minimal order $n-\text{rank}\begin{bmatrix}C&F\end{bmatrix}+\text{rank}(F)$, if the following assumption holds.
\begin{assumption}\label{assp1}
The system matrices satisfy the following two conditions:
\begin{itemize}
\item [$1^\circ$] $\text{rank}\begin{bmatrix}0&F\\F&CE\end{bmatrix}
=\text{rank}(F)+\text{rank}\begin{bmatrix}E\\F\end{bmatrix}$;
\item [$2^\circ$] $\text{rank}\begin{bmatrix}A-sI_n&E\\C&F\end{bmatrix}
=n+\text{rank}\begin{bmatrix}E\\F\end{bmatrix},~\forall s\in\mathbf{C}$.
\end{itemize}
\end{assumption}
\begin{remark}
It is revealed in \cite{houm1994tac} that $(G\bar A,\bar C)$ is detectable if and only if condition $2^\circ$ in Assumption \ref{assp1} holds for all $s$ in non-negative real part. For the case that $\text{rank}(\bar C)=n$, i.e., $\text{rank}(C)=n$ and $F=0$, the state $x$ can be directly calculated by measurement output $y$ as $x=(C^TC)^{-1}C^T(y-Du)$, which appears to be the trivial state estimation case. In this paper, we only consider the case when $\text{rank}(\bar C)<n$.
\end{remark}

Based on the pairwise observer structure of \cite{engel2002tac}, we can formulate the following observer:
\begin{equation}\label{observer00}
\begin{aligned}
\dot{\zeta}&=\bar A_c\zeta+\bar H_cy+\bar B_cu,\\
\hat{\eta}(t)&=\bar D_c[\zeta(t)-e^{\bar A_c\tau}\zeta(t-\tau)],\\
\bar{x}&=\hat\eta+\bar E(\bar C\bar E)^+(I_m-FF^+)(y-Du),
\end{aligned}
\end{equation}
where $\bar A_c=\text{diag}(G\bar A+K_1\bar C,G\bar A+K_2\bar C)$ with $K_1$ and $K_2$ being gain matrices such that $\Re_j \{\lambda(G\bar A+K_2\bar C)\}<\sigma<\Re_k\{\lambda(G\bar A+K_1\bar C)\}<0,\forall j,k=1,\cdots,n$, $\bar H_c=\begin{bmatrix}\bar H_{c1}\\ \bar H_{c2}\end{bmatrix}$ and $\bar B_c=\begin{bmatrix}\bar B_{c1}\\ \bar B_{c2}\end{bmatrix}$ with $\bar H_{ci}=H-K_i\hat C,\bar B_{ci}=\bar B+K_i\hat CD,i=1,2$, and $\bar D_c=\begin{bmatrix}I_n&0\end{bmatrix}\begin{bmatrix}\bar C_c&e^{\bar A_c\tau}\bar C_c\end{bmatrix}^{-1}$. Clearly, $\hat\eta(t)\equiv \eta(t),\forall t\geq\tau$, and we have the following result.
\begin{corollary}\label{cor1}
Suppose that Assumption \ref{assp1} holds. For almost all $\tau>0$, the $2n$-order observer \dref{observer00} can estimate the state $x(t)$ of the linear system \dref{model} in presence of the unknown input $w$ at appointed time $\tau$.
\end{corollary}

\subsection{Reduced-order Observer Design}
The reduced-order appointed-time unknown input observer was designed in \cite{rafft2006mcca} for the linear system \dref{model} with $F=0$. This subsection intends to present the corresponding observer when $F\neq 0$.

Decompose $\bar E$ into $\bar E=\bar E_0\bar E_1$ with $\bar E_0\in\mathbf{R}^{n\times \text{rank}(\bar E)}$ being of full column rank. Then the model \dref{model1} is rewritten as
\begin{equation}\label{model10}
\begin{aligned}
\dot{x}&=\bar Ax+EF^+y+\hat Bu+\bar E_0\bar E_1w,\\
\bar y&=\bar Cx.
\end{aligned}
\end{equation}
Choose $\bar T_0{\in}\mathbf{R}^{n{\times} (n{-}\text{rank}(\bar E))}$ such that $\bar T{=}\begin{bmatrix}\bar T_0&\bar E_0\end{bmatrix}$ is invertible. Let $\bar T^{-1}{=}\begin{bmatrix}\bar T_1\\ \bar T_2\end{bmatrix}$ with $\bar T_1{\in}\mathbf{R}^{(n{-}\text{rank}(\bar E)){\times} n}$ and $\bar T_2{\in}\mathbf{R}^{\text{rank}(\bar E){\times} n}$. Then, we have $\bar T_1\bar E_0=0$. Further choose $\bar U_0\in\mathbf{R}^{m\times (m-\text{rank}(\bar E))}$ such that $\bar U=\begin{bmatrix}\bar U_0&\bar C\bar E_0\end{bmatrix}$ is invertible, and let $\bar U^{-1}=\begin{bmatrix}\bar U_1\\ \bar U_2\end{bmatrix}$. The existence of $\bar U$ can be derived by the following result.
\begin{theorem}\label{lem4}
Under Assumption \ref{assp1}, $\bar C\bar E_0$ is of full column rank, i.e.,
$$\text{rank}(\bar C\bar E_0)=\text{rank}(\bar E).$$
\end{theorem}

Before proceeding, we first introduce the following lemma.
\begin{lemma}[\cite{houm1994tac}]\label{lem10}
For any matrices $M$ and $N$ with appropriate dimensions, $N(I-M^+M)=0$ if and only if $\text{rank}\begin{bmatrix}N\\M\end{bmatrix}=\text{rank}(M)$.
\end{lemma}

Now we are ready to demonstrate Theorem \ref{lem4}.

\prooft
Note that
$$\bar E(I_q-(\bar C\bar E)^+\bar C\bar E)=G\bar E=0.$$
In light of Lemma \ref{lem10}, we have
$$\text{rank}(\bar E)\leq\text{rank}\begin{bmatrix}\bar E\\ \bar C\bar E\end{bmatrix}=\text{rank}(\bar C\bar E)
\leq\text{rank}(\bar E).$$
Thus, the statement that $\bar C\bar E_0$ is of full column rank can be derived by noticing
$$\text{rank}(\bar E)=\text{rank}(\bar C\bar E)\leq\text{rank}(\bar C\bar E_0)
\leq\text{rank}(\bar E_0). ~~~~~~~\hfill\blacksquare$$

Let $\psi=\bar T_1x$. We have $x=(I_n-\bar E_0\bar U_2\bar C)\bar T_0\psi+\bar E_0\bar U_2\bar y=(I_n-\bar E_0\bar U_2\bar C)\bar T_0\psi+\bar E_0\bar U_2(I_m-FF^+)(y-Du)$. The dynamics of $\psi$ can be described by
\begin{equation}\label{model200}
\begin{aligned}
\dot{\psi}&=\bar T_1\bar A(I_n-\bar E_0\bar U_2\bar C)\bar T_0\psi+\hat Hy+\hat{\bar{B}}u,\\
y_{\psi}&=\bar U_1\bar C\bar T_0\psi=\bar U_1\bar y,
\end{aligned}
\end{equation}
where $\hat{\bar{B}}=\bar T_1\hat B-\bar T_1\bar A\bar E_0\bar U_2(I_m-FF^+)D$ and
$\hat H=\bar T_1EF^++\bar T_1\bar A\bar E_0\bar U_2(I_m-FF^+)$.

The reduced-order appointed-time unknown input observer can be designed by following the observer structure in \cite{rafft2006mcca} as follows.
\begin{equation}\label{observer01}
\begin{aligned}
&\dot{\hat{\zeta}}=\hat A_c\zeta+\hat H_cy+\hat B_cu,\\
&\hat{\psi}(t)=\hat D_c[\hat\zeta(t)-e^{\hat A_c\tau}\hat\zeta(t-\tau)],\\
&\hat{\bar{x}}{=}(I_n{-}\bar E_0\bar U_2\bar C)\bar T_0\hat{\psi}
{+}\bar E_0\bar U_2(I_m{-}FF^+)(y{-}Du),
\end{aligned}
\end{equation}
where $\hat A_c=\text{diag}(\hat A_{c1},\hat A_{c2})$, $\hat A_{ci}=\bar T_1\bar A(I_n-\bar E_0\bar U_2\bar C)\bar T_0+\hat K_i\bar U_1\bar C\bar T_0$ with $\hat K_i$ being the gain matrices such that $\Re_j \{\lambda(\hat A_{c2})\}<\sigma<\Re_k\{\lambda(\hat A_{c1})\}<0,j,k=1,\cdots,n-\text{rank}(\bar E)$, $\hat H_c=\begin{bmatrix}\hat H_{c1}\\ \hat H_{c2}\end{bmatrix}$ and $\hat B_c=\begin{bmatrix}\hat B_{c1}\\ \hat B_{c2}\end{bmatrix}$ with $\hat H_{ci}=\hat H-\hat K_i\bar U_1(I_m-FF^+),\bar B_{ci}=\hat B+\hat K_i\bar U_1(I_m-FF^+)D,i=1,2$, $\hat C_c=\begin{bmatrix}I_{n-\text{rank}(\bar E)}\\I_{n-\text{rank}(\bar E)}\end{bmatrix}$, and $\hat D_c=\begin{bmatrix}I_{n-\text{rank}(\bar E)}&0\end{bmatrix}\begin{bmatrix}\hat C_c&e^{\hat A_c\tau}\hat C_c\end{bmatrix}^{-1}$. It was revealed in \cite{rafft2006mcca} that $\hat{\psi}(t)\equiv \psi(t),\forall t\geq\tau$, and we have the following result.
\begin{corollary}\label{cor2}
Suppose that Assumption \ref{assp1} holds. For almost all $\tau>0$, the $2\bigg(n-\text{rank}\begin{bmatrix}E\\F\end{bmatrix}+\text{rank}(F)\bigg)$-order observer \dref{observer01} can estimate the state $x(t)$ of the linear system \dref{model} in presence of the unknown input $w$ at appointed time $\tau$.
\end{corollary}
\begin{remark}
It should be pointed out that the order reduction of observer \dref{observer01}, compared with full-order observer \dref{observer00}, is mainly due to the different model reconfigurations. Specifically, to remove the effect of unknown input, model \dref{model2} adopts the variable $\eta$ with order $n$, and model \dref{model200} uses variable $\psi$ with order $n-\text{rank}\begin{bmatrix}E\\F\end{bmatrix}+\text{rank}(F)$. While the design structure and the principle of appointed-time estimation realization of the full-order observer \dref{observer00} and the reduced-order observer \dref{observer01} are the same.
\end{remark}

\subsection{The Gap Between Reduced-order and Minimal-order Observers}
Since the minimal-order asymptotical observer is of order $n-\text{rank}\begin{bmatrix}C&F\end{bmatrix}+\text{rank}(F)$, the minimal-order appointed-time unknown input observer is expected to be of order $2\left(n-\text{rank}\begin{bmatrix}C&F\end{bmatrix}+\text{rank}(F)\right)$.
The following result shows that the expected minimal order is strictly smaller than the order of observer \dref{observer01}.
\begin{theorem}\label{lem5}
Suppose that Assumption \ref{assp1} holds, and $\text{rank}(\bar C)<n$. Then, $$\text{rank}\begin{bmatrix}C&F\end{bmatrix}>\text{rank}\begin{bmatrix}E\\F\end{bmatrix}.$$
\end{theorem}

Before proceeding, we first introduce the following lemma, the proof of which is omitted since it is a dual result of Lemma \ref{lem10}.
\begin{lemma}\label{lem100}
For any matrices $M$ and $N$ with appropriate dimensions, $(I-MM^+)N=0$ if and only if $\text{rank}\begin{bmatrix}N&M\end{bmatrix}=\text{rank}(M)$.
\end{lemma}

Now we are ready to show the proof of Theorem \ref{lem5}.

\proofthm
Since $\text{rank}(\bar C)+\text{rank}(F)=\text{rank}\begin{bmatrix}C&F\end{bmatrix}$ and
$\text{rank}(\bar E)+\text{rank}(F)=\text{rank}\begin{bmatrix}E\\F\end{bmatrix}$, it is equivalent to showing that $\text{rank}(\bar C)>\text{rank}(\bar E)$.

By Theorem \ref{lem4}, we have $\text{rank}(\bar C)\geq\text{rank}(\bar E)$. We then prove that the equality cannot hold by contradiction. Assume that $\text{rank}(\bar C)=\text{rank}(\bar E)$, which implies that $\text{rank}\begin{bmatrix}\bar C&\bar C\bar E\end{bmatrix}=\text{rank}(\bar C\bar E)$. By Lemma \ref{lem100}, we have $0=(I-\bar C\bar E(\bar C\bar E)^+)\bar C=\bar CG$. For $s\neq0$,
\begin{equation*}
\begin{aligned}
n=&\text{rank}\begin{bmatrix}sI_n-G\bar A\\ \bar C\end{bmatrix}
=\text{rank}\begin{bmatrix}I_n&0\\ \bar C&-sI_m\end{bmatrix}\begin{bmatrix}sI_n-G\bar A\\ \bar C\end{bmatrix}\\
=&\text{rank}\begin{bmatrix}sI_n-G\bar A\\ 0\end{bmatrix},
\end{aligned}
\end{equation*}
implying that $G\bar A$ has only zero eigenvalues. Since $\text{rank}(\bar C)<n$, we can obtain that $G\bar A\neq0$. Let $\Xi=\text{diag}(\Xi_i)=\Upsilon^{-1}G\bar A\Upsilon$ be the Jordan canonical form of $G\bar A$ with $\Xi_1\neq0$. Note that
\begin{equation*}
\begin{aligned}
n{=}\text{rank}\begin{bmatrix}-G\bar A\\ \bar C\end{bmatrix}
{=}\text{rank}\begin{bmatrix}\Upsilon^{-1}&0\\0& I_m\end{bmatrix}
\begin{bmatrix}G\bar A\\ \bar C\end{bmatrix}\Upsilon
{=}\text{rank}\begin{bmatrix}\Xi\\ \bar C\Upsilon\end{bmatrix},
\end{aligned}
\end{equation*}
implying that the first column of $\bar C\Upsilon$ cannot be $0$. Then the second column of $\bar C\Upsilon\Xi$ is nonzero. However, $\bar C\Upsilon\Xi=\bar CG\bar A\Upsilon=0$. This completes the proof. $\hfill\blacksquare$

In view of Theorem \ref{lem5}, the order of the appointed-time observer can be further reduced.
In the following subsection, we intend to design the appointed-time unknown input observer of the minimal order $2\left(n-\text{rank}\begin{bmatrix}C&F\end{bmatrix}+\text{rank}(F)\right)$.

\subsection{Minimal-order Observer Design}
Since $\text{rank}(\bar C)=\text{rank}\begin{bmatrix}C&F\end{bmatrix}-\text{rank}(F)$, $\bar C$ can never be of full row rank if $F\neq0$.
Let $\bar C_0$ be the matrix with its rows consisting of a maximal linearly independent set of the rows of $\bar C$. Denote $\bar C_0=\bar C_1\bar C$. Clearly, under Assumption \ref{assp1}, $(G\bar A,\bar C_0)$ is observable. Let $y_{\eta_0}=\bar C_0\eta$. The dynamics of $\eta$ can be rewritten as
\begin{equation}\label{model3}
\begin{aligned}
\dot{\eta}&=G\bar A\eta+Hy+\bar Bu,\\
y_{\eta_0}&=\bar C_0\eta=\bar C_1\hat C(y-Du),
\end{aligned}
\end{equation}
Following the design steps of minimal-order appointed-time observer \dref{obs12}, we can formulate the appointed-time observer $\bar\eta(t)$ to estimate $\eta(t)$ at appointed time $\tau$ as follows:
\begin{equation}\label{observerz1}
\begin{aligned}
\dot{z}_1&=\mathcal{M}_1z_1+\mathcal{N}_1y+\hat {\mathcal{N}}_1u,\\
\dot{z}_2&=\mathcal{M}_2z_2+\mathcal{N}_2y+\hat {\mathcal{N}}_2u,\\
\bar\eta(t)&=\mathcal{D}\begin{bmatrix}\mathcal{U}_1&~\\~&\mathcal{U}_2\end{bmatrix}
\left[\bar\phi(t)-e^{\widehat{\mathcal{M}}\tau}\bar\phi(t-\tau)\right],
\end{aligned}
\end{equation}
where $z_i\in\mathbf{R}^{n-\text{rank}\begin{bmatrix}C&F\end{bmatrix}+\text{rank}(F)}$ are the observer states, $\bar\phi=[z_1^T,y_{\eta_0}^T,z_2^T,y_{\eta_0}^T]^T$,
$\mathcal{M}_i$ are stable matrices which have no common eigenvalues with $G\bar A$ and $\Re_j\{
\lambda(\mathcal{M}_2)\}<\sigma<\Re_k\{\lambda(\mathcal{M}_1)\}<0,j,k=1,\cdots,n-m$,
$\widehat{\mathcal{M}}=\begin{bmatrix}\widehat{\mathcal{M}}_1&~\\~&\widehat{\mathcal{M}}_2
\end{bmatrix}$ with $\widehat{\mathcal{M}}_i=\begin{bmatrix}{\mathcal{M}}_i&~\\~&\bar{{M}}_i\end{bmatrix}$,
$\mathcal{N}_i=\mathcal{H}_i\bar C_1\hat C+\mathcal{T}_iH$, $\hat{\mathcal{N}}_i=\mathcal{T}_i\bar B-\mathcal{H}_i\bar C_1\hat CD$, $\mathcal{H}_i$ are the gain matrices such that $(\mathcal{M}_i,\mathcal{H}_i)$ are controllable, $\mathcal{T}_i$ are the unique solutions of the Sylvester equations
\begin{equation}\label{sylvester}
\mathcal{T}_iG\bar A-\mathcal{M}_i\mathcal{T}_i=\mathcal{H}_i\bar C_0
\end{equation}
such that the matrices $\begin{bmatrix}\mathcal{T}_i\\ \bar C_0\end{bmatrix}$ are invertible with $\mathcal{U}_i=\begin{bmatrix}\mathcal{T}_i\\ \bar C_0\end{bmatrix}^{-1}$, and ${\mathcal{D}}=\begin{bmatrix}I_n&0\end{bmatrix}\begin{bmatrix}I_n&{\mathcal{U}}_1
e^{\widehat{\mathcal{M}}_1\tau} {\mathcal{U}}_1^{-1}\\I_n&{\mathcal{U}}_2e^{\widehat{\mathcal{M}}_2\tau} {\mathcal{U}}_2^{-1}\end{bmatrix}^{-1}$.

Then the minimal-order appointed-time unknown input observer is constructed by
\begin{equation}\label{observer}
\begin{aligned}
\hat{x}&=\bar\eta+\bar E(\bar C\bar E)^+(I_m-FF^+)(y-Du).
\end{aligned}
\end{equation}
\begin{theorem}\label{thmphi}
Suppose that Assumption \ref{assp1} holds. For almost all $\tau>0$, the minimal-order unknown input observer $\hat x$ presented in \dref{observer} can estimate the state $x(t)$ of the linear system \dref{model} in presence of the unknown input $w$ at appointed time $\tau$.
\end{theorem}

\begin{remark}
The common feature of the full-order observer \dref{observer00} and the reduced-order observer \dref{observer01} is that the two observer states are designed to asymptotically estimate the same variable. Specifically, in the full-order observer \dref{observer00}, both $\zeta_1$ and $\zeta_2$ are to estimate $\eta$; while in the reduced-order observer \dref{observer01}, both $\hat\zeta_1$ and $\hat\zeta_2$ are to estimate $\psi$. In the minimal-order observer however, two observer states, namely $z_1$ and $z_2$, are introduced to estimate different variables ($\mathcal{T}_1\eta$ and $\mathcal{T}_2\eta$, respectively).
Consequently, we cannot construct any appropriate variable to be estimated at appointed-time $\tau$ with only $z_i(t)$ and the delayed information $z_i(t-\tau)$. Instead, the variable $\eta$ is chosen to be estimated at appointed time $\tau$. To generate the observer $\bar\eta$, not only the information of observer states $z_i$, but also the output information $y_{\eta_0}$ at both time instants $t$ and $t-\tau$ are used.
\end{remark}

For the special case $F=0$, Assumption \ref{assp1} degenerates into the following assumption.
\begin{assumption}\label{assp2}
The system matrices satisfy the following two conditions:
\begin{itemize}
\item [$1^\circ.$] $\text{rank}(CE)
=\text{rank}(E)$;
\item [$2^\circ.$] $\text{rank}\begin{bmatrix}sI_n-A&E\\C&0\end{bmatrix}
=n+\text{rank}(E),~\forall s\in\mathbf{C}$.
\end{itemize}
\end{assumption}

Let $\breve G=I_n-E(CE)^+C$. It is revealed in \cite{lvyz2020tcns} that Assumption \ref{assp2} is equivalent to that $\breve GE=0$ and $(\breve GA,C)$ is observable. Without loss of generality, we assume that $C$ is of full row rank. Define $\breve\eta=\breve Gx$. Denote $y_{\breve{\eta}}=C\breve\eta$. We have
\begin{equation*}
\begin{aligned}
\dot{\breve{\eta}}&=\breve GA\breve\eta+\breve GAE(CE)^+y+(\breve GB-\breve GAE(CE)^+D)u,\\
y_{\breve{\eta}}&=C\breve\eta=(I_m-CE(CE)^+)(y-Du).
\end{aligned}
\end{equation*}
And the minimal-order appointed-time unknown input observer for the linear system \dref{model} with $F=0$ would be
\begin{equation}\label{observerz0}
\begin{aligned}
\dot{\breve{z}}_1&=\breve{\mathcal{M}}_1\breve z_1+\breve{\mathcal{N}}_1y+\breve{\hat {\mathcal{N}}}_1u,\\
\dot{\breve{z}}_2&=\breve{\mathcal{M}}_2\breve z_2+\breve{\mathcal{N}}_2y+\breve{\hat {\mathcal{N}}}_2u,\\
\breve{\bar{\eta}}(t)&=\breve{\mathcal{D}}\begin{bmatrix}\breve{\mathcal{U}}_1&~\\~&\breve{\mathcal{U}}_2\end{bmatrix}
\left[\breve{\bar{\phi}}(t)-e^{\check{\mathcal{M}}\tau}\breve{\bar{\phi}}(t-\tau)\right],\\
\breve{x}&=\breve{\bar{\eta}}+ E(CE)^+(y-Du),
\end{aligned}
\end{equation}
where $\breve z_i{\in}\mathbf{R}^{n{-}m}$ are observer states, $\breve{\bar{\phi}}{=}[\breve z_1^T,y_{\breve\eta}^T,\breve z_2^T,y_{\breve\eta}^T]^T$,
$\breve{\mathcal{M}}_i$ are stable matrices which have no common eigenvalues with $\breve GA$ and $\Re_j\{
\lambda(\breve{\mathcal{M}}_2)\}<\sigma<\Re_k\{\lambda(\breve{\mathcal{M}}_1)\}{<}0$, $j,k
=1,\cdots,n{-}m$,
$\check{\mathcal{M}}{=}\begin{bmatrix}\check{\mathcal{M}}_1&~\\~&\check{\mathcal{M}}_2\end{bmatrix}$, $\check{\mathcal{M}}_i{=}\begin{bmatrix}{\breve{\mathcal{M}}}_i&~\\~&\bar{{M}}_i\end{bmatrix}$,
$\breve{\mathcal{N}}_i=\breve{\mathcal{H}}_i(I_m-CE(CE)^+)+\breve{\mathcal{T}}_i\breve GAE(CE)^+$, $\breve{\hat{\mathcal{N}}}_i=\breve{\mathcal{T}}_i(\breve GB-\breve GAE(CE)^+D)-\breve{\mathcal{H}}_i(I_m-CE(CE)^+)D$, $\breve{\mathcal{H}}_i$ are gain matrices such that $(\breve{\mathcal{M}}_i,\breve{\mathcal{H}}_i)$ are controllable, $\breve{\mathcal{T}}_i$ are the unique solutions of the Sylvester equations
\begin{equation*}
\breve{\mathcal{T}}_i\breve G A-\breve{\mathcal{M}}_i\breve{\mathcal{T}}_i=\breve{\mathcal{H}}_i C
\end{equation*}
such that the matrices $\begin{bmatrix}\breve{\mathcal{T}}_i\\  C\end{bmatrix}$ are invertible with $\breve{\mathcal{U}}_i=\begin{bmatrix}\breve{\mathcal{T}}_i\\ C\end{bmatrix}^{-1}$, and $\breve{\mathcal{D}}=\begin{bmatrix}I_n&0\end{bmatrix}\begin{bmatrix}I_n&\breve{\mathcal{U}}_1
e^{\check{\mathcal{M}}_1\tau} \breve{\mathcal{U}}_1^{-1}\\I_n&\breve{\mathcal{U}}_2e^{\check{\mathcal{M}}_2\tau} \breve{\mathcal{U}}_2^{-1}\end{bmatrix}^{-1}$.
\begin{corollary}
Suppose that Assumption \ref{assp2} holds. For almost all $\tau>0$, the minimal-order unknown input observer $\breve x$ presented in \dref{observerz0} can estimate the state $x(t)$ of the linear system \dref{model} with $F=0$ at appointed time $\tau$.
\end{corollary}

\section{Applications on Attack-free Protocol Design for Multi-agent Systems}\label{s4}
In this section, we intend to apply the minimal-order appointed-time unknown input observer \dref{observerz0} into the consensus problem of multi-agent systems.

Consider the linear multi-agent system consisting of $N$ agents, whose dynamics are given by
\begin{equation}\label{modelm}
\begin{aligned}
&\dot{x}_i=\breve Ax_i+\breve Bu_i,\\
&y_i=\breve Cx_i,~i=1,\cdots,N,
\end{aligned}
\end{equation}
where $x_i\in\mathbf{R}^n$, $y_i\in\mathbf{R}^m$ and $u_i\in\mathbf{R}^p$ are respectively the state, the output and the control input of the $i$-th agent.

The communication topology among the $N$ agents is described by a directed graph $\mathcal{G}=\{\mathcal{V},\mathcal{E}\}$, where $\mathcal{V}=\{1,\cdots,N\}$ is the node set and $\mathcal{E}\subset\mathcal{V}\times\mathcal{V}$ is the edge set. The adjacency matrix $\mathcal{A}=[a_{ij}]$ is defined by $a_{ij}=1$ if $(i,j)\in\mathcal{E}$, and $a_{ij}=0$ otherwise. The Laplacian matrix $\mathcal{L}=[l_{ij}]$ is defined by $l_{ii}=\sum_{j=1}^Na_{ij}$, and $l_{ij}=-a_{ij}$ for $j\neq i$.

In this section, we aim at designing fully distributed adaptive attack-free output-feedback protocol to realize consensus of the $N$ agents in \dref{modelm}, which is formulated in our previous work \cite{lvyz2020auto}.
\begin{problem}[Attack-free protocols]\label{problem}
Design fully distributed output-feedback protocol in the form
\begin{equation}\label{pro}
\begin{aligned}
\dot{\chi}_i=&f_i\left(\sum_{j=1}^Na_{ij}(y_i-y_j),\chi_i\right),\\
u_i=&g_i\left(\sum_{j=1}^Na_{ij}(y_i-y_j),\chi_i\right),
\end{aligned}
\end{equation}
with $f_i$ and $g_i$ as the nonlinear functions such that $\lim_{t\rightarrow \infty}\|x_i(t)- x_j(t)\|=0$,
$\forall i,j$.
\end{problem}

The protocol in Problem \ref{problem} possesses the feature that the observer information exchange among neighboring agents via communication channel is not required, and only relative output $y_i-y_j$ measured by local sensors is used. In this manner, the information transmission via network is removed, which has the advantage of saving communication cost and makes the protocol free from network attacks.

Define the consensus error as $\xi_i=\sum_{j=1}^Na_{ij}(x_i-x_j)$, and the consensus is realized if and only if $\xi_i$ converges to zero. \cite{lvyz2020tcns} designed the attack-free protocols
by introducing the distributed full-order and reduced-order appointed-time unknown input observers to estimate the consensus error, where the relative control input is viewed as the unknown input. The following assumption is made in \cite{lvyz2020tcns}.
\begin{assumption}\label{assp3}
The triple $(\breve A,\breve B,\breve C)$ satisfies:
\begin{itemize}
\item [$1^\circ$] $\text{rank}(\breve C\breve B)
=\text{rank}(\breve B)=p$;
\item [$2^\circ$] $\text{rank}\begin{bmatrix}sI_n-\breve A&\breve B\\ \breve C&0\end{bmatrix}
=n,~\forall s\in\mathbf{C}$.
\end{itemize}
\end{assumption}

Based on the minimal-order appointed-time unknown input observer \dref{observerz0}, the following distributed adaptive attack-free protocol is proposed:
\begin{equation}\label{protocol}
\begin{aligned}
\dot{\varpi}_{i1}=&\bar{\mathcal{M}}_1{\varpi}_{i1}+\bar{\mathcal{N}}_1\sum_{j=1}^Na_{ij}(y_i-y_j),\\
\dot{\varpi}_{i2}=&\bar{\mathcal{M}}_2{\varpi}_{i2}+\bar{\mathcal{N}}_2\sum_{j=1}^Na_{ij}(y_i-y_j),\\
\bar \xi_i(t)=&\bar{\mathcal{D}}\begin{bmatrix}\bar{\mathcal{U}}_1&~\\~&\bar{\mathcal{U}}_2\end{bmatrix}
\begin{bmatrix}\bar\varpi_{i1}(t)-e^{\widetilde{\mathcal{M}}_1\tau}\bar\varpi_{i1}(t-\tau)\\
\bar\varpi_{i2}(t)-e^{\widetilde{\mathcal{M}}_2\tau}\bar\varpi_{i2}(t-\tau)\end{bmatrix},\\
\hat\xi_i(t)=&\bar\xi_i+\breve B(\breve C\breve B)^+\sum_{j=1}^Na_{ij}(y_i-y_j),\\
u_i=&\begin{cases}0,~~~~~~~~~~~~~~~~~~~~~~~~~~~~~~~t<\tau,\\
-(\rho_i+\hat\xi_i^TP\hat\xi_i)\breve B^TP\hat\xi_i,~~t\geq\tau,\end{cases}\\
\dot{\rho}_i=&\begin{cases}0,~~~~~~~~~~~~~~~~~~~t<\tau,\\
\hat\xi_i^TP\breve B\breve B^TP\hat\xi_i,~~t\geq\tau,\end{cases}
\end{aligned}
\end{equation}
where $\bar\varpi_{ik}=[\varpi_{ik}^T,[(I_m-\breve C\breve B(\breve C\breve B)^+)\sum_{j=1}^Na_{ij}(y_i-y_j)]^T]^T$, $\bar{\mathcal{M}}_k$ are stable matrices which have no common eigenvalues with $\hat GA$ and satisfy $\Re_{l_2}\{\lambda(\bar{\mathcal{M}}_2)\}<\sigma<\Re_{l_1}\{\lambda(\bar{\mathcal{M}}_1)\}<0,
l_1,l_2=1,\cdots,n-m$, $\hat G=I_n-\breve B(\breve C\breve B)^+\breve C$, $\bar{\mathcal{N}}_k=\bar{\mathcal{T}}_k\hat GA\breve B(\breve C\breve B)^++\bar{\mathcal{H}}_k(I_m-\breve C\breve B(\breve C\breve B)^+)$,
$\bar{\mathcal{H}}_k$ are gain matrices such that $(\bar{\mathcal{M}}_k,\bar{\mathcal{H}}_k)$ are controllable, $\bar{\mathcal{U}}_k
=\begin{bmatrix}\bar{\mathcal{T}}_k\\ \breve C\end{bmatrix}^{-1}$ with $\bar{\mathcal{T}}_k$ being the unique solution to the Sylvester equation
\begin{equation}\label{sylvester10}
\bar{\mathcal{T}}_k\hat G\breve A-\bar{\mathcal{M}}_k\bar{\mathcal{T}}_k=\bar{\mathcal{H}}_k \breve C,
\end{equation}
$\widetilde{\mathcal{M}}_k=\text{diag}(\bar{\mathcal{M}}_k,\bar{M}_k)$, $\bar{\mathcal{D}}=\begin{bmatrix}I_n&0\end{bmatrix}\begin{bmatrix}I_n&\bar{\mathcal{U}}_1
e^{\widetilde{\mathcal{M}}_1\tau} \bar{\mathcal{U}}_1^{-1}\\I_n&\bar{\mathcal{U}}_2e^{\widetilde{\mathcal{M}}_2\tau} \bar{\mathcal{U}}_2^{-1}\end{bmatrix}^{-1}$, $\rho_i$ is the adaptive gain with initial value $\rho_i>0$, and $P$ is a positive definite matrix whose inverse $Q=P^{-1}$ is the solution to the LMI:
\begin{equation}\label{lmi}
AQ+QA^T-2BB^T<0.
\end{equation}

It is clear that $\hat\xi_i(t)\equiv\xi_i(t),\forall t\geq\tau$. And we have the following result in light of the consensus realization under the fully distributed adaptive state-feedback protocol [\cite{lv2022scis}].
\begin{corollary}
Suppose that Assumption \ref{assp3} holds, and the communication topology $\mathcal{G}$ contains a directed spanning tree. The consensus of the $N$ agents in \dref{modelm} can be achieved under the fully distributed adaptive attack-free output-feedback protocol \dref{protocol}.
\end{corollary}
\begin{remark}
Compared with the protocols presented in \cite{lvyz2020tcns}, the distributed adaptive attack-free protocol \dref{protocol} has a lower order and saves the computation cost, as it
adopts the minimal-order observer to estimate the exact value of consensus error at the appointed time. It should be clarified that, to realize appointed-time estimation of the consensus error with a relatively small $\tau$, the overshooting of $\hat\xi_i$ during the time period $[0,\tau)$ is surely to be extremely large, which in turn results in the over-threshold of the control input due to the limited control ability.
To overcome the above limitation existing in \cite{lvyz2020tcns}, zero input is introduced during the time period $[0,\tau)$ and the protocol \dref{protocol} is a piecewise controller. Meanwhile, the consensus error would not be far from the initial value since $\tau$ can be chosen arbitrarily small, which helps avoid input saturation.
\end{remark}
\begin{remark}
It should be noticed that in the construction of protocol \dref{protocol}, the design of the minimal-order appointed-time observer and the control input is decoupled. Consequently, the minimal-order appointed-time observer \dref{observer} can be easily implemented into the consensus and formation problems for multi-agent systems with uncertainties and nonlinearities, or the attack detection and identification of cyber-physical systems.
\end{remark}

\begin{example}
For a multi-agent system consisting of six agents whose dynamics are described by \dref{modelm} with
\begin{equation*}
\begin{aligned}
&\breve A=\begin{bmatrix}
0   &  1   &  0\\
     1  &  -1 &    1\\
     0  &  -8  &   1
\end{bmatrix},~
\breve B=\begin{bmatrix}0\\0\\1\end{bmatrix},
\breve C=\begin{bmatrix}1&0&0\\0&0&1\end{bmatrix},
\end{aligned}
\end{equation*}
the communication graph is depicted in Fig. \ref{fig1}, which clearly contains a directed spanning tree.

\begin{figure}[!h]
\begin{center}
\includegraphics[width=1.8in]{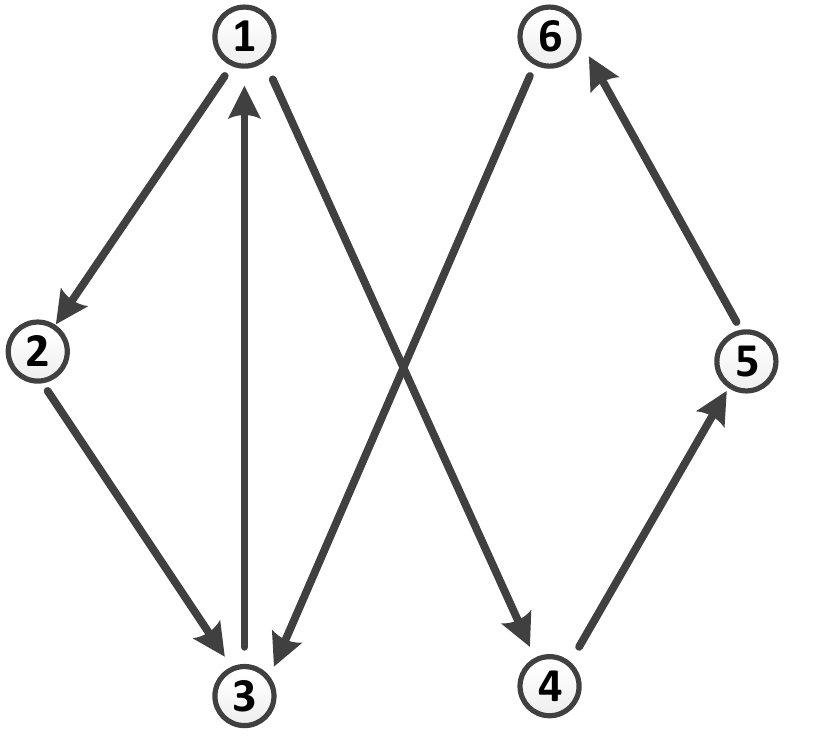}    
\caption{The directed communication graph.}  
\label{fig1}                                 
\end{center}                                 
\end{figure}

We then have $\hat G=\begin{bmatrix}    1  & 0&0\\  0&1&0\\
      0&0&0\end{bmatrix}$ and $\hat G\breve A=\begin{bmatrix}    0   &  1  &   0\\
     1  &  -1 &    1\\
      0&0&0\end{bmatrix}$. The eigenvalues of $\hat G\breve A$ are $0,0.618,-1.618$, and we can choose $\bar{\mathcal{M}}_1=-1,\bar{\mathcal{M}}_2=-2$ with $\widetilde{\mathcal{M}}_1=-I_3$ and $\widetilde{\mathcal{M}}_2=-2I_3$. Let $\tau=1$. Choose $\bar{\mathcal{H}}_1=\bar{\mathcal{H}}_2=\begin{bmatrix}    1 & 0 \end{bmatrix}$. Solving the Sylvester equation \dref{sylvester10} gives $\bar{\mathcal{T}}_1=\begin{bmatrix}    0 & 1&-1 \end{bmatrix}$ and $\bar{\mathcal{T}}_2=\begin{bmatrix}    1 & -1&0.5 \end{bmatrix}$. We can calculate $\bar{\mathcal{N}}_1=\begin{bmatrix}1&1\end{bmatrix}$, $\bar{\mathcal{N}}_2=\begin{bmatrix}1&-1\end{bmatrix}$, $\bar{\mathcal{U}}_1=\begin{bmatrix}0&1&0\\1&0&1\\0 &  0&1\end{bmatrix}$, $\bar{\mathcal{U}}_2=\begin{bmatrix}0&1&0\\-1&1&0.5\\0 &  0&1\end{bmatrix}$ and $\bar{\mathcal{D}}=\begin{bmatrix}-0.582I_3&1.582I_3\end{bmatrix}$.
Solving the LMI \dref{lmi} gives a solution
$Q=\begin{bmatrix}
   0.8695  & -0.1369   &-1.1761\\
   -0.1369  &  0.2512   & 0.3033\\
   -1.1761  &  0.3033   & 2.9821
\end{bmatrix}$. Then,
$P=\begin{bmatrix}
  2.4730    &0.1941 &   0.9555\\
    0.1941  &  4.5540 &  -0.3867\\
    0.9555  & -0.3867  &  0.7515
\end{bmatrix}$.
The initial states of the agents and observers are randomly chosen, and $\rho_i(0)=1$.

Fig. \ref{fig31} shows that the minimal-order appointed-time unknown input observer $\hat\xi_i$ can exactly estimate the consensus error $\xi_i$ at appointed time $\tau$.
The consensus error $\xi_i$ is depicted in Fig. \ref{figxi1} and the state $x_i$ is presented in Fig. \ref{figx1}. It is clear that the consensus is indeed realized though the state of each agent diverges.
\begin{figure}
\centering
\includegraphics[width=3.5in]{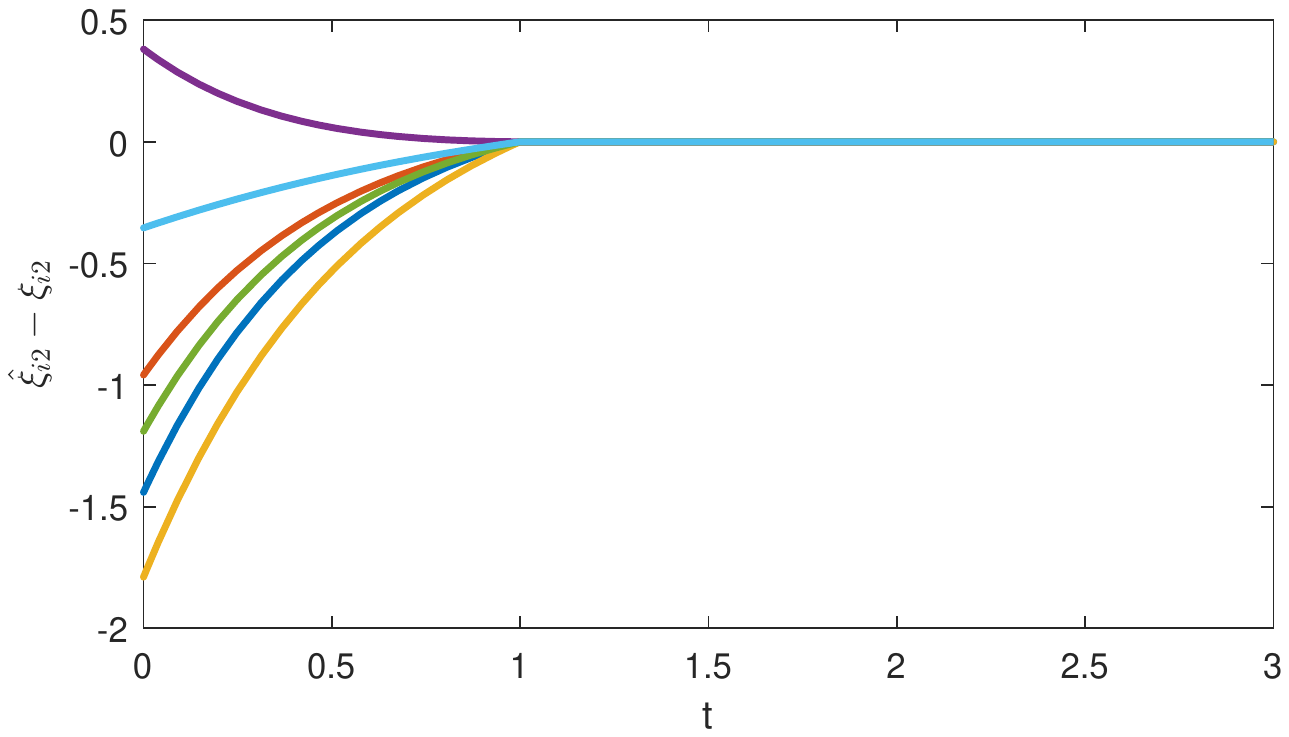}
\caption{The observer $\hat\xi_{i2}$ estimating the consensus error $\xi_{i2}$ at appointed time $\tau$.}\label{fig31}
\end{figure}
\begin{figure}
  \centering
  \subfigure{
    \label{fig:subfig:a11} 
    \includegraphics[width=3.5in]{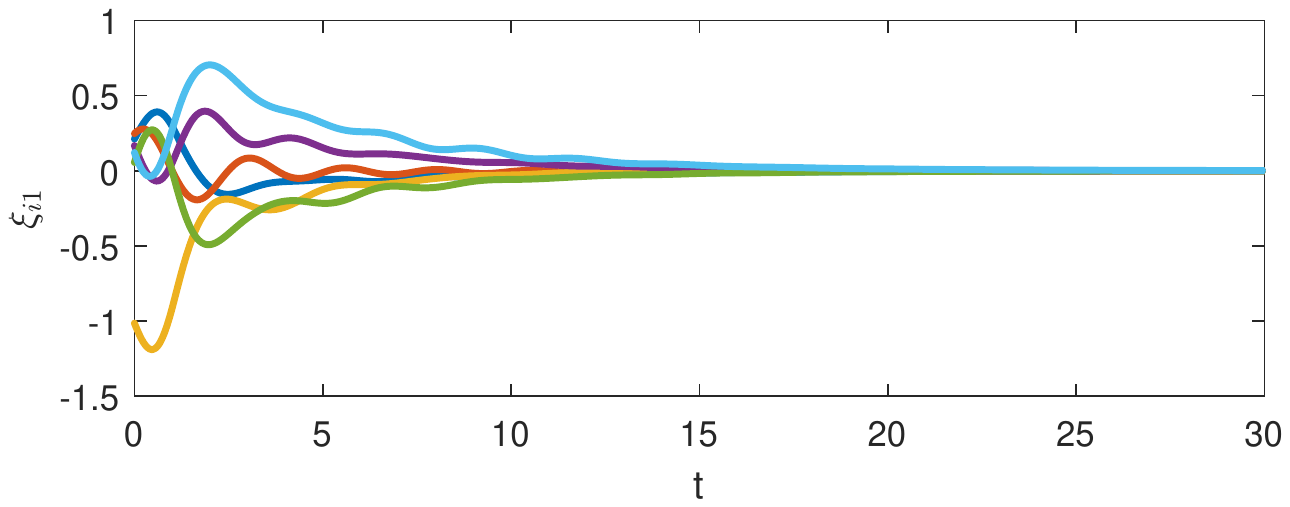}}
  \hspace{0.01in}
  \subfigure{
    \label{fig:subfig:b11} 
    \includegraphics[width=3.5in]{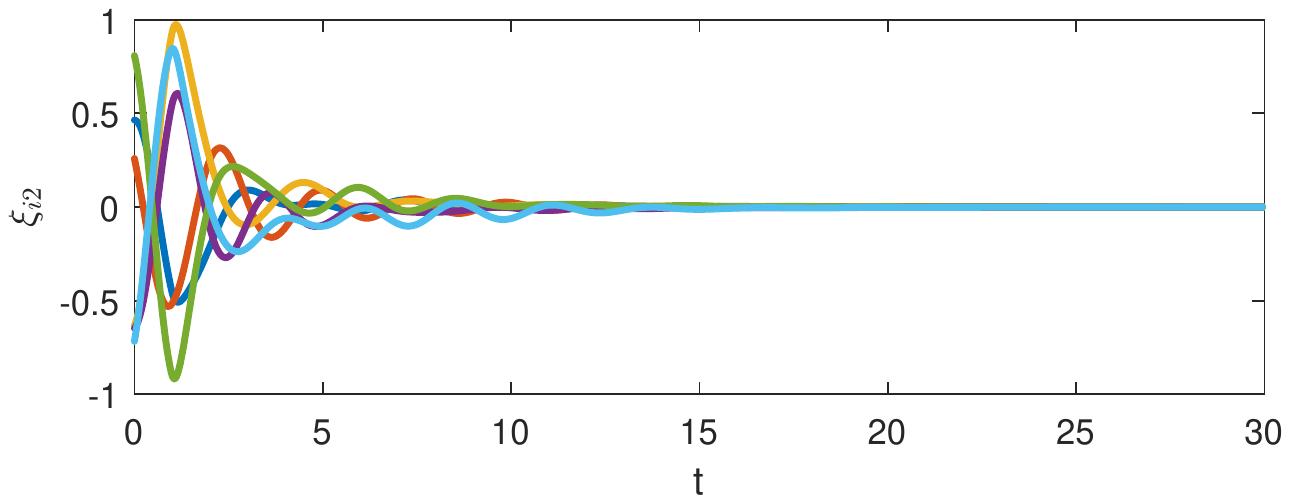}}
  \hspace{0.01in}
  \subfigure{
    \label{fig:subfig:c11} 
    \includegraphics[width=3.5in]{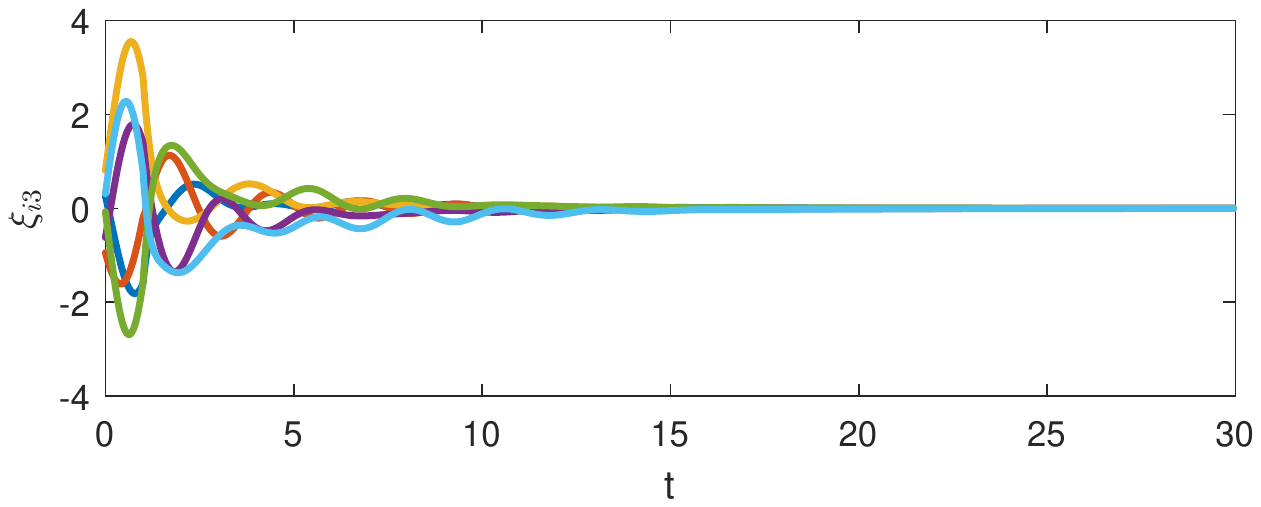}}
  \caption{The consensus error $\xi$ converging to zero.}
  \label{figxi1}
\end{figure}
\begin{figure}
  \centering
    \subfigure{
    \label{fig:subfig:a0} 
    \includegraphics[width=3.5in]{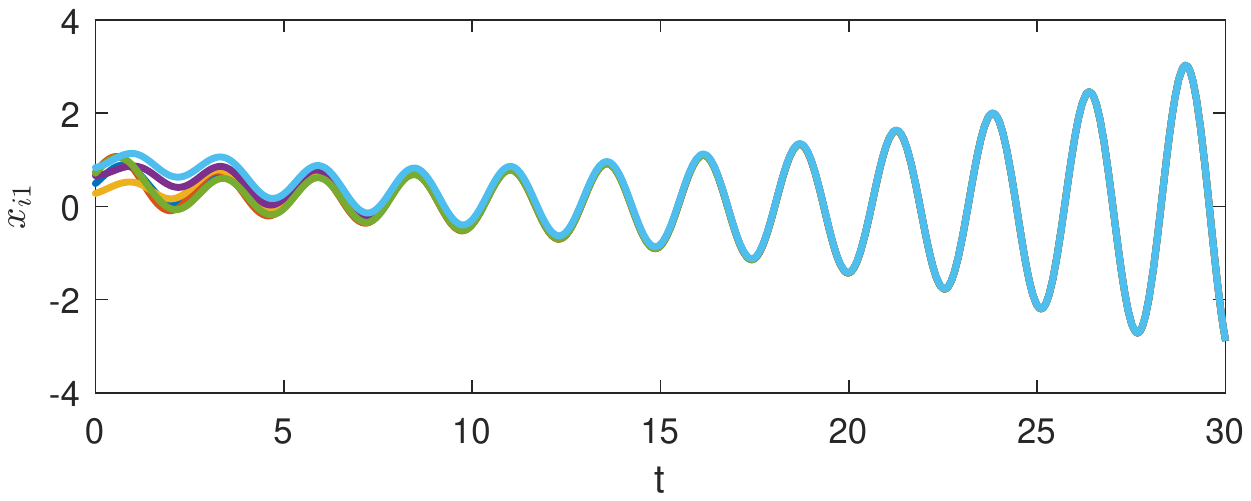}}
  \hspace{0.01in}
  \subfigure{
    \label{fig:subfig:b0} 
    \includegraphics[width=3.5in]{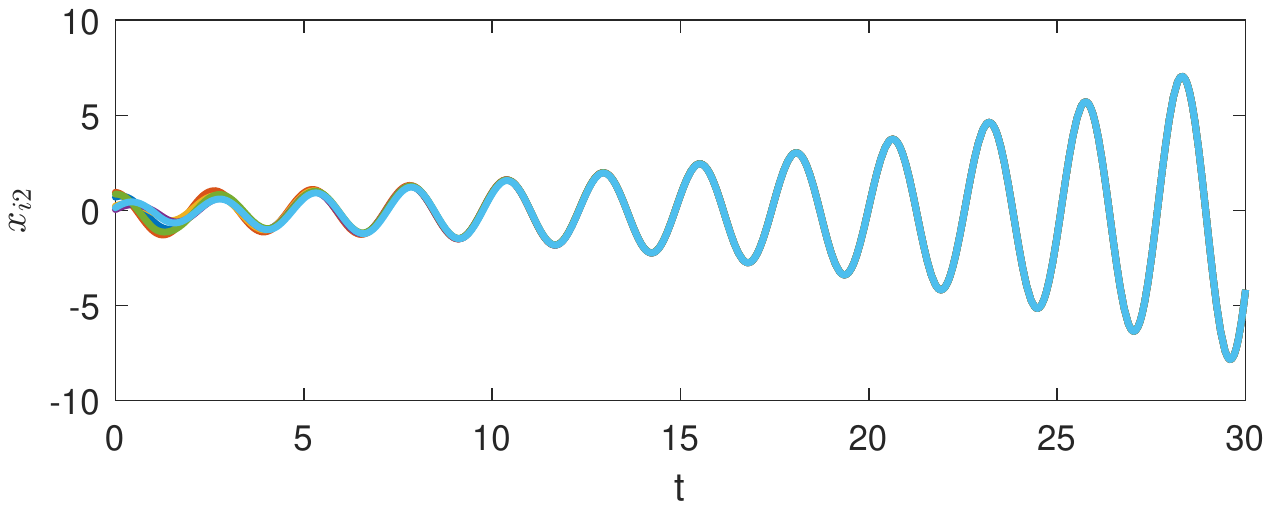}}
  \hspace{0.01in}
  \subfigure{
    \label{fig:subfig:c0} 
    \includegraphics[width=3.5in]{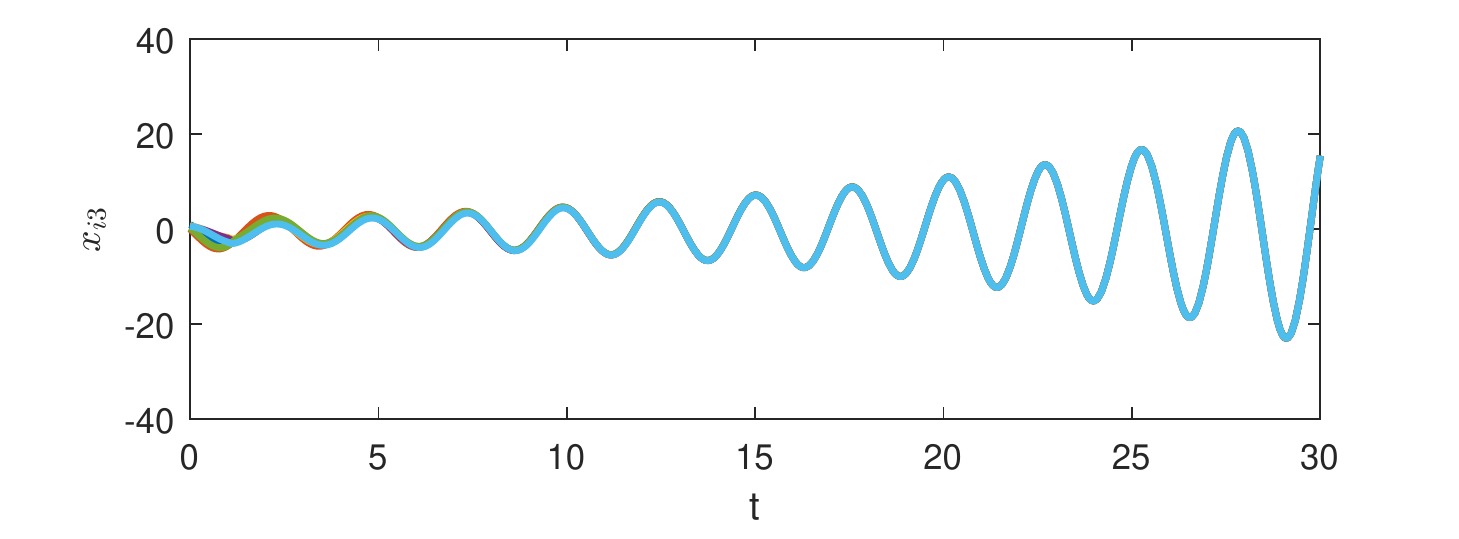}}
  \caption{The state $x$ reaching consensus.}
  \label{figx1}
\end{figure}

\end{example}

\section{Conclusion}\label{s6}
In this paper, we have presented a unified framework of the minimal-order appointed-time observer design based on the pairwise observer structure for linear systems. It has been revealed that the structure of the proposed observer is coincident with that of the full-order appointed-time observer in \cite{engel2002tac}. The general model of linear system with the unknown input was also considered, and the minimal-order appointed-time unknown input observer was then proposed, which has lower order than the appointed-time observers in existing literature.
The minimal-order appointed-time observer design methodology can be easily extended into the functional observer design for linear systems or nonlinear systems, and can be further applied to the attack detection for cyber-physical systems.


%

\end{document}